\documentclass[11pt]{article}
\usepackage{amssymb}
\usepackage{amsthm}
\usepackage{latexsym}
\usepackage{amsmath}
\usepackage{graphicx}
\usepackage{tikz}
\usepackage[shortlabels]{enumitem}
\usepackage{tabularx}
\usepackage{booktabs}
\usepackage{makecell}
\numberwithin{equation}{section}
\newtheorem{theorem}{Theorem}[section]
\newtheorem{cor}[theorem]{Corollary}
\newtheorem{lemma}[theorem]{Lemma}
\newtheorem{prop}[theorem]{Proposition}

\theoremstyle{definition}
\newtheorem{defn}[theorem]{Definition}
\newtheorem{ex}[theorem]{Example}
\newtheorem{algorithm}[theorem]{Algorithm}
\newtheorem{remark}[theorem]{Remark}

\newcommand{\N}{{\mathbb N}}
\newcommand{\F}{{\mathbb F}}

\newcommand{\Z}{{\mathbb Z}}

\newcommand{\R}{{\mathbb R}}

\newcommand{\AND}{{\rm AND}}
\newcommand{\OR}{{\rm OR}}
\newcommand{\twoby}[4]{\begin{pmatrix}#1&#2\\#3&#4\end{pmatrix}}

\newcommand{\rankSum}{\text{RankSum}}
\newcommand{\sbt}{\raisebox{.2ex}{\mbox{$\scriptscriptstyle\bullet\,$}}}
\newcommand{\rank}{\mbox{${\rm rk}$}}

\renewcommand{\mod}{\mbox{\rm mod}\,}

\newcommand{\wt}{{\rm wt}}
\newcommand{\wtH}{\mbox{{\rm wt}$_{\rm H}$}}
\newcommand{\wtL}{\mbox{{\rm wt}$_{\rm L}$}}
\newcommand{\wtE}{\mbox{{\rm wt}$_{\rm E}$}}
\newcommand{\wtU}{\mbox{{\rm wt}$_{\rm U}$}}

\newcounter{alp}
\newcounter{ara}
\newcounter{rom}
\newenvironment{romanlist}{\begin{list}{(\roman{rom})\hfill}{\usecounter{rom}
     \topsep0ex \labelwidth.7cm \leftmargin.7cm \labelsep0cm
     \rightmargin0cm \parsep0ex \itemsep.4ex
     \partopsep1ex}}{\end{list}}
\newenvironment{alphalist}{\begin{list}{(\alph{alp})\hfill}{\usecounter{alp}
     \topsep0ex \labelwidth.6cm \leftmargin.6cm \labelsep0cm
     \rightmargin0cm \parsep0ex \itemsep0ex
     \partopsep0ex}}{\end{list}}
\newenvironment{arabiclist}{\begin{list}{\arabic{ara}.\hfill}{\usecounter{ara}
     \topsep0ex \labelwidth.6cm \leftmargin.6cm \labelsep0cm
     \rightmargin0cm \parsep0ex \itemsep0ex
     \partopsep1.6ex}}{\end{list}}

\begin{document}
%%%%%%%%%%%%%%%%%%%%%%%%%%%%%%%%
\title{Lexicodes over Finite Principal Left Ideal Rings}
\date{\today}
\author{Jared Antrobus$^\ast$ and Heide Gluesing-Luerssen\footnote{HGL was partially supported by the National Science Foundation Grant
  DMS-1210061 and by the grant \#422479 from the Simons Foundation. 
  HGL and JA are with the Department of Mathematics, University of Kentucky, Lexington KY 40506-0027, USA;
\{jantrobus,heide.gl\}@uky.edu.}}

\maketitle

{\bf Abstract:}

Let $R$ be a finite principal left ideal ring.
Via a total ordering of the ring elements and an ordered basis a lexicographic ordering of the module~$R^n$ is produced.
This is used to set up a greedy algorithm that selects vectors for which all linear combination with the previously selected vectors satisfy a
pre-specified selection property and updates the to-be-constructed code to the linear hull of the vectors selected so far.
The output is called a lexicode.
This process was discussed earlier in the literature for fields and chain rings.
In this paper we investigate the properties of such lexicodes over finite principal left ideal rings
and show that the total ordering of the ring elements has to respect containment of ideals
in order for the algorithm to produce meaningful results.
Only then it is guaranteed that the algorithm is exhaustive and thus produces codes that are maximal with respect to inclusion.
It is further illustrated that the output of the algorithm heavily depends on the total ordering and chosen basis.


\section{Introduction}
Lexicodes, or lexicographic codes, were first introduced by Levenstein~\cite{Lev60} in 1960 with the goal to construct binary codes with a desired minimal Hamming distance. They are obtained by ordering all binary vectors lexicographically and applying a greedy algorithm that selects the vectors that have at least the desired Hamming distance from all previously selected vectors.
Interestingly, the resulting codes turn out to be linear.
Later in 1986, Conway/Sloane~\cite{CoSl86} generalized the idea to codes over fields of characteristic~$2$.
Focusing primarily on codewords realized as winning positions in game theory, they showed that the resulting lexicodes are always additive,
and they are linear if the field size is $2^{2^k}$ for some $k\geq0$.
Many well-known codes, such as the Hamming codes and the extended binary Golay code, turn out to be lexicodes; for a brief overview see~\cite{CoSl86}.

In all the above cases the vectors of the search space~$\F^n$ are ordered by suitably interpreting them as binary representation of integers.
In 1993, Brualdi/Pless~\cite{BrPl93} generalized the theory to using arbitrary ordered bases of $\F_2^n$ and ordering the space by using the lexicographic
ordering on the coefficient vectors.
Among other things, they proved that the resulting codes are again linear.

In 1997 this result has been further generalized by Van Zanten~\cite{VZ97} by allowing other selection criteria instead of the Hamming distance.
More precisely,  Van Zanten presented the following simple algorithm for constructing codes satisfying some property $P$ over the lexicographically ordered space~$\F_2^n$:
\begin{align*}
  &\text{Denote the vectors selected so far by~$C$.}\\
  &\text{Select the next vector $x$ in the list $\F_2^n$ such that $P[x+y]$ holds true for all $y\in C$.}\\
  &\text{Update $C$ to $C\cup\{x\}$.}
 \end{align*}
As in the earlier cases where the property~$P$ was a desired minimum Hamming distance, it turns out that  the resulting code is linear~\cite{VZ97}.
The result is generalized to codes over fields of characteristic~$2$ and, again, linearity is established if the field is of size $2^{2^k}$ for some~$k$ and
the field elements are ordered suitably.

In 2005, a shift in the construction of lexicodes occurred by imposing linearity of the code via an adjustment of the greedy algorithm.
In~\cite{VZS05} Van Zanten/Suparta considered the search space $\F^n$ for general fields~$\F$ and ordered it into level sets based on an
ordered basis along with some fixed, yet arbitrary, ordering of the field elements.
Choosing a selection property~$P$ on~$\F^n$ that is invariant under scalar multiplication, they set up the following greedy algorithm:
\begin{align*}
  &\text{Denote the vectors selected so far by~$C$.}\\
  &\text{Go to the next level set and find the first vector~$x$ such that $P[x+y]$ is true for all $y\in C$.}\\
  &\text{Update~$C$ to $C+\F x$.}
\end{align*}
The resulting lexicode is clearly linear.
However, in this variant it is not a priori clear whether all added vectors $\alpha x+y,\alpha\in\F,y\in C$, satisfy the selection property.
Fortunately, this is indeed the case as established in~\cite{VZS05}.
Another interesting feature of the algorithm is that each level set is searched only once: if the search is successful the algorithm moves on to the next level set after its update.
It is proved in~\cite{VZS05} that the algorithm is nevertheless exhaustive in that it does not miss any admissible vectors.

In 2014, Guenda et al.~\cite{GGS14} generalize the results from~\cite{VZS05} to codes over commutative chain rings~$R$.
In that case, the selection property for $R^n$ has to be invariant under multiplication by units.
Moreover, the test for $P[x+y]$ in the above algorithm needs to be replaced by $P[\gamma^j x+y]$ for all~$j$, where~$\gamma$ is a
generator of the maximal ideal.

In this paper we revisit the results of~\cite{GGS14} and extend them to codes over finite principal left ideal rings.
In this case a code (of length~$n$) is a left submodule of~$R^n$.
As in~\cite{GGS14} we consider selection properties that are invariant under multiplication by units.
Only this guarantees meaningful results of the greedy algorithm.
The algorithm is essentially as the above one with~$\F$ replaced by~$R$ in the update, and with $P[x+y]$ replaced by $P[\gamma x+y]$,
where~$\gamma$ runs through a set of generators of the nonzero left ideals of~$R$.
While these are the obvious generalizations of the chain ring case, special attention needs to be paid to the ordering of the space~$R^n$.
Again it is based on an ordered basis along with an ordering of the ring elements.
However, the latter one needs to be chosen with care for the greedy algorithm to produce good results.
More precisely, the ordering of the ring has to respect containment of (nonzero) left ideals, see Definition~\ref{def:Respectful}.
Only then it is guaranteed that the algorithm is exhaustive and the resulting codes are maximal within the set of all
codes satisfying the given property.
The exhaustiveness is nontrivial and proven with the aid of the stable range property of principal left ideal rings.
Even though the same stipulations on the ordering of the ring also apply to chain rings, this has not been addressed explicitly in~\cite{GGS14}.
This may be due to the fact that many chain rings, such as $\Z_{p^r}:=\Z/_{p^r\Z}$ for any prime~$p$ and other small chain rings,
come with a `natural' order, which seems to have been tacitly assumed in~\cite{GGS14}.
These orderings do indeed respect containments of ideals.

An interesting role is played by the value of the selection property for the zero vector.
It is not hard to see that the lexicode is free if the zero vector does not satisfy the selection property.
However, even though we may easily toggle the value of the property for the zero vector between true and false, the outcome of the
greedy algorithm may fundamentally change. This is illustrated by various examples in Section~\ref{SS-Exa}.
In addition, the lexicode heavily depends on the ordering of the ring elements (even if the ordering respects ideal containment).
This is also true in the field case where even the dimension of the lexicode may depend on the ordering.
In Section~\ref{SS-Exa}  we present an abundance of examples illustrating the various features of the algorithm and, in particular,
the dependence of the lexicode on the ordering.

The paper is organized as follows.
In the next section we recall crucial properties of finite principal left ideal rings and discuss various weight functions as well as
other properties that may serve as selection property for a greedy algorithm.
In Section~\ref{SS-Ord} we introduce respectful orderings on~$R$ and establish their existence.
We use such an ordering along with an ordered basis of the left $R$-module~$R^n$ to order the module lexicographically.
Section~\ref{SS-Alg} is devoted to the greedy algorithm and its properties.
Finally, in Section~\ref{SS-Exa}  we present examples illustrating the various features of the algorithm and the dependence of
the lexicode on the ordering.

%%%%%%%%%%%%%%%%%%%%%%%%%%%%%%%%%%
\section{Preliminaries}
%%%%%%%%%%%%%%
We begin with some basic ring-theoretic properties that will be needed later on.
For now let~$R$ denote any (non-commutative) ring with identity.
We use the notation $R^\ast$ for the group of units of~$R$.

We need to collect some crucial properties of finite principal left ideal rings and start with the stable range.
A a ring~$R$ is said to have \emph{(left) stable range 1} if whenever $p,q\in R$ satisfy $Rp+Rq=R$,
there exists $t\in R$ such that $tp+q\in R^\ast$; see \cite[(20.10)]{Lam01}.
Right stable range 1 is defined similarly.
In \cite[Thm.~1.8]{Lam04} Lam shows that left and right stable range 1 are actually equivalent properties.

Recall that a ring~$R$ is called \emph{semilocal} if $R/\text{rad}(R)$ is left artinian, where
rad$(R)$ denote the Jacobson radical of~$R$. Clearly, all finite rings are semilocal.
The following result is known as Bass' Theorem.

%%%%%%%%%%%%%%%%%%%%%%%%%%%
\begin{theorem}[\mbox{\cite[(20.9)]{Lam01}}]
Let $R$ be a semilocal ring, $q\in R$, and $I$ a left ideal of $R$. If $I+Rq=R$, then the coset $I+q$ contains a unit of~$R$.
Thus~$R$ has stable range~$1$.
In particular, every finite ring has stable range~$1$.
\end{theorem}
%%%%%%%%%%%%%%%%%%%%%%%%%%%%%%

The next result provides a useful characterization of rings with stable range 1.

%%%%%%%%%%%%%%%%%%%%%%%%%%%%%%%
\begin{theorem}[\mbox{\cite[Thm.~1.9]{Lam04} or \cite[Thm.~2.9]{Can95}}]\label{thm:UnitLinearComb}
Let~$R$ be any ring. The following are equivalent.
\begin{romanlist}
\item $R$ has stable range 1.
\item If $p,q,d\in R$ satisfy $Rp+Rq=Rd$, then there exists $t\in R$ and $u\in R^\ast$ such that $tp+q=ud$.
\end{romanlist}
\end{theorem}
%%%%%%%%%%%%%%%%%%%%%%%%%%%%%%

We now turn to codes over~$R$. The following definition is standard.
Throughout, all modules are left $R$-modules.
%%%%%%%%%%%%%%%%%%%%%%%
\begin{defn}\label{D-Code}
Let $n\in\N$. A \emph{code of length~$n$ over the alphabet~$R$} is a left submodule of $R^n$.
\end{defn}
%%%%%%%%%%%%%%%%%%%%%%%

Bass' Theorem leads to a well-known and extremely useful consequence.

%%%%%%%%%%%%%%%%%%%%%%%%%%%
\begin{prop}[\mbox{\cite[Prop.~5.1]{Wo99}}]\label{prop:AssociateGenerators}
Let $R$ be any finite ring and~$M$ a left $R$-module. Let $a,b\in M$ be such that $Ra=Rb$. Then $ua=b$ for some $u\in R^\ast$.
\end{prop}
%%%%%%%%%%%%%%%%%%%%%%%%%%%%%%
Note that if~$R$ has stable range~$1$, then Proposition~\ref{prop:AssociateGenerators} follows immediately for the module
$M=R$ since $R0+Rb=Ra$ implies $b=ua$ for
some $u\in R^\ast$ thanks to Theorem~\ref{thm:UnitLinearComb}.
In fact, \cite[Theorem 1.9(3)]{Lam04} shows that for the case $M=R$ the property in Proposition~\ref{prop:AssociateGenerators}
characterizes stable range~$1$.

The next corollary follows trivially.

\begin{cor}\label{cor:UnitOrbits}
Let~$R$ and~$M$ be as in Proposition~\ref{prop:AssociateGenerators}.
Then the group $R^\ast$ acts naturally on $M$ by $(u,a)\mapsto ua$.
The orbits of this group action are exactly the sets of generators for the distinct cyclic left submodules of~$M$.
In particular, the orbits of the action of~$R^*$ on~$R$ are the sets of generators for the distinct principal left ideals of~$R$.
\end{cor}

In this paper we focus on codes over finite principal left ideal rings.
Recall that a ring is called a \emph{principal left ideal ring} if every left ideal is principal.
In \cite[p.~364]{Nec73} Nechaev showed that every finite principal left ideal ring is a principal ideal ring
(that is, each left ideal and each right ideal is principal).
One may notice that finite principal left ideal rings are Frobenius rings because they have a principal left socle, see~\cite[Thm.~1]{Hon01}.

For modules over such rings we have the following powerful property.

%%%%%%%%%%%%%%%%%%%%%%%%%%
\begin{theorem}\label{T-FLPIR}
Let~$R$ be a finite principal left ideal ring and let $N,M$ be free left $R$-modules of finite rank such that $M$ is a submodule of $N$.
Then $M$ is a direct summand of $N$, that is, there is a submodule~$P$ of~$N$ such that $M\oplus P=N$.
\end{theorem}
%%%%%%%%%%%%%%%%%%%%%%%%%%

\begin{proof}
By \cite[p.~364/365]{Nec73} each each finite principal left ideal ring is the direct sum of matrix rings over finite chain rings.
Now the result follows from \cite[Thm.~4.7]{Hi16} by Hirano who proved that the desired direct summand property is true for all rings that are
direct sums of matrix rings over finite local rings.
\end{proof}

A special case of finite principal left ideal rings are finite chain rings.
Recall that a \emph{finite chain ring} is a finite ring wherein the left ideals are linearly ordered with respect to inclusion.
It turns out that the left ideals of a finite chain ring~$R$ are all two-sided and therefore agree with the right ideals.
In fact, $R$ can be characterized as a local ring whose maximal ideal is principal and generated by some nilpotent element $\gamma\in R$.
If~$e$ is the nilpotency index of~$\gamma$ then the ideals of~$R$ are given by the chain
\begin{equation}\label{e-idealChain}
R=(1)
\supsetneq (\gamma)
\supsetneq (\gamma^2)
\supsetneq \ldots
\supsetneq (\gamma^{e-1})
\supsetneq (\gamma^e)
=(0).
\end{equation}
For all this, see, for instance, ~\cite[Thm.~2.1]{HoLa00} by Honold/Landjev and the references therein.

\medskip
We now turn to various coding-theoretic weight functions.
Let~$R$ be any finite ring.
A map $w: R\longrightarrow\R$ satisfying $w(0)=0$ is called a \emph{weight function} on $R$.
Any such weight~$w$ has a natural extension to vectors $(x_1,\ldots,x_n)\in R^n$ via the rule
\begin{equation}\label{e-ExtWeight}
   w(x_1,\ldots,x_n)=\sum_{i=1}^n w(x_i).
\end{equation}
Here are some special instances of weight functions.

%%%%%%%%%%%%%%%%%%%%%%%%%%%%%%
\begin{defn}\label{D-Weights}
Let $R$ be a ring.
\begin{alphalist}
\item The \emph{Hamming weight} $\wtH$ on $R$ is defined by the rule $\wtH(0)=0$ and $\wtH(x)=1$ for all $x\in R\,\backslash\,\{0\}$.
\item Let $R=M_k(\F)$, the ring of $k\times k$-matrices over the finite field~$\F$.
        We define the rank weight of~$X\in R$ as the rank of~$X$, denoted by $\rank(X)$.
        For a vector $x=(X_1,\ldots,X_n)\in M_k(\F_q)^n$ we define the \emph{rank sum}  as in~\eqref{e-ExtWeight}
\[
    \rankSum(x)=\sum_{i=1}^n\rank(X_i).
\]
\item On any finite ring~$R$ set $\wtU(a)=1$ if $a\in R^*$ and $\wtU(a)=0$ otherwise.
        Then $\wtU(x_1,\ldots,x_n)$ counts the number of units in the vector $(x_1,\ldots,x_n)$.
\item On $R=\Z_m:=\Z/m\Z$ the \emph{Lee weight} is defined as $\wtL(x)=\min(x,m-x)$ and the
      \emph{Euclidean weight} is defined as $\wtE(x)=\min(x,m-x)^2$.
\end{alphalist}
\end{defn}
%%%%%%%%%%%%%%%%%%%%%%%%%%%%

In addition to the above, the homogeneous weight plays a prominent role in ring-linear coding.
The following definition is taken from \cite[Definition 1.2]{GrSch00} by Greferath/Schmidt.
In the same paper the authors also establish existence and uniqueness of the homogeneous weight for all finite rings.

%%%%%%%%%%%%%%%%%%%%%%%%%%%%%%%%%%%%%
\begin{defn}\label{D-HomogWeight}
Let $R$ be a finite ring. A function $\omega:R\to\R$ is called the (normalized left) \emph{homogeneous weight} if it satisfies the following properties.
\begin{romanlist}
\item $\omega(0)=0$
\item If $Ra=Rb$ for $a,b\in R$, then $\omega(a)=\omega(b)$.
\item For every $a\in R$ we have $\sum_{x\in Ra}\omega(x)=|Ra|.$
\end{romanlist}
\end{defn}
%%%%%%%%%%%%%%%%%%%%%%%%%%%%%%%%%%%%%%

%%%%%%%%%%%%%%%%%%%%%%%%%%%%%%%%%
\begin{ex}
On~$\Z_2$ and~$\Z_3$ the Hamming weight and Lee weight agree, and the homogeneous weight agrees with these up to a factor~$2$.
On~$\Z_4$, the normalized homogeneous weight agrees with the Lee weight and is given by the values
$\omega(0)=0,\ \omega(1)=\omega(3)=1,\ \omega(2)=2$.
On~$\Z_m$, where $m>4$, the Hamming weight, Lee weight, and homogeneous weight are mutually distinct.
\end{ex}
%%%%%%%%%%%%%%%%%%%%%%%%%%%%%%%

In the next sections we will discuss a greedy algorithm that results in codes having a pre-specified property.
The property serves as the selection criterion in the algorithm.
We will need the property to be multiplicative in the following sense.

%%%%%%%%%%%%%%%%%%%%%%%%%%%%%%
\begin{defn}
Let $R$ be a ring.
A boolean function $P:R^n\longrightarrow \{\text{true}, \text{false}\}$ is called a \emph{property} on $R^n$.
We call~$P$ \emph{left multiplicative} if $P[x]=P[u x]$ for all $u\in R^\ast$.
\emph{Right multiplicative} is defined analogously.
If~$P$ is both left and right multiplicative, we simply call $P$ \emph{multiplicative}.
\end{defn}
%%%%%%%%%%%%%%%%%%%%%%%%%%%%%%%%
Often we will simply write $P[x]$ for $P[x]=\text{true}$.
For instance, $[P[x]\Longleftrightarrow\wtH(x)>\delta]$ means that $P[x]=\text{true}$ if $\wtH(x)>\delta$ and false otherwise.

Many selection properties may be desirable in order to construct codes. The following are some commonly desired properties.
%%%%%%%%%%%%%%%%%%%%%%%%%%%%%%
\begin{ex}\label{ex:properties}
(a) Let~$w$ be any of the weights introduced in Definition~\ref{D-Weights}(a) --~(c), Definition~\ref{D-HomogWeight} or the Lee
   or Euclidean weight on~$\Z_4$ (see~\ref{D-Weights}(d)).
      Extend~$w$ to~$R^n$ as in~\eqref{e-ExtWeight}. In all these cases $w(ux)=w(x)$ for $x\in R^n$ and $u\in R^*$.
      Therefore, for any $\delta\in\R$, the property $[P[x]\Longleftrightarrow w(x)\geq\delta]$ is left multiplicative.
      The same is true for the property $[P[x]\Longleftrightarrow w(x)\in S]$, where $S$ is a pre-specified set of
      admissible weight values (such as even weights).
      In particular,  $[P[x]\Longleftrightarrow\rankSum(x)\geq\delta]$ is a multiplicative property on $M_k(\F)^n$ for any $\delta>0$.

(b) Let~$R$ be any commutative ring and denote by $x\cdot y:=\sum_{i=1}^n x_i y_i$ the standard dot product on~$R^n$.
      Then the property $[P[x]\Longleftrightarrow x\cdot x=0]$ is multiplicative because for any $u\in R^*$ we have
       $(u x)\cdot(u x)=u^2(x\cdot x)$.
       The same property is in general not multiplicative if~$R$ is not commutative (as one easily verifies for the matrix ring $M_2(\F_2)$).

(c) Let~$I\subseteq R$ be a left ideal of~$R$. On~$R^n$ define $[P[x]\Longleftrightarrow \sum_{i=1}^n x_i\in I]$.
        Then~$P$ is left multiplicative.
\end{ex}
%%%%%%%%%%%%%%%%%%%%%%%

Of course, there are plenty of other multiplicative properties over finite rings.
For example, the sum of the entries being a unit is a multiplicative property.
However, this property is not useful for our purposes.
Indeed, we will aim at constructing \emph{linear} codes with a desired property, and thus in order to obtain non-trivial codes
we need the property to be reasonably conserved upon multiplication by arbitrary ring elements.
Similarly, the property that the sum of the entries is a zero divisor (even though preserved by multiplication with any ring element)
will often not lead to codes with more than one generator as this property is scarcely preserved under addition.

One particular property can, for many rings, be used to construct self-orthogonal codes.
Let us summarize the necessary information about self-orthogonal codes.

%%%%%%%%%%%%%%%%%%%%%%%%%%%
\begin{remark}\label{rem:Isotr}
Let~$R$ be a commutative ring.
On $R^n$ consider the (multiplicative) property $[P[x]\Longleftrightarrow x\cdot x =0]$, where~$x\cdot y$ denotes the
standard dot product, see Example~\ref{ex:properties}(b).
If the characteristic of~$R$ is odd then a linear code $C\subseteq R^n$ satisfies
\[
  x\cdot x=0\text{ for all } x\in C \Longrightarrow x\cdot y=0\text{ for all }x,y\in C.
\]
This follows immediately from $0=(x+y)\cdot(x+y)=x\cdot x+2(x\cdot y)+y\cdot y=2(x\cdot y)$, and since~$2$ is not a zero divisor,
we obtain the desired result.
Recall that the dual code of~$C\subseteq R^n$ is defined as $C^{\perp}:=\{y\in R^n\mid y\cdot x=0\text{ for all }x\in C\}$
and that~$C$ is \emph{self-orthogonal} (resp.\ \emph{self-dual}) if $C\subseteq C^\perp$ (resp. $C=C^\perp$).
The above shows that if the characteristic of~$R$ is odd, then a code~$C$ is self-orthogonal if
$x\cdot x=0$ for all $x\in C$.
Self-orthogonality is thus characterized by a suitable property for the individual elements of the code (instead of pairs of elements).
Finally, we remark that if~$R$ is a finite principal left ideal ring, and thus in particular a Frobenius ring,
and $C\subseteq R^n$ a code, then $|C|\cdot|C^\perp|=|R^n|$.
This is a consequence of the double annihilator property for finite Frobenius rings; see for instance~\cite[p.~193]{GL14homog}.
\end{remark}
%%%%%%%%%%%%%%%%%%%%%%%%%

We close this section with addressing the value $P[0]$ for a given property~$P$.
This will play an interesting role in Section~\ref{SS-Alg}.

Most standard properties are not satisfied by the zero vector; for instance $P[0]$ is false for the very common criterion
[$P[x]\Longleftrightarrow\wtH(x)\geq\delta$], where~$\delta>0$.
We can easily set the value of $P[0]$ to our liking, due to the following proposition and corollary.

%%%%%%%%%%%%%%%%%%%%%%%%%%%%%%%%%%%
\begin{prop}\label{prop:LogicOperators}
The family of (left) multiplicative properties is closed under the logical operators  \AND\ and \OR.
\end{prop}
%%%%%%%%%%%%%%%%%%%%%%%%%%%%%%%%
\begin{proof}
Suppose $P$ and $Q$ are both left multiplicative properties on $R^n$. We then have, for any $\beta\in R^\ast$ and $x\in R^n$,
$(P\;\AND\;Q)[\beta x]=(P[\beta x]\;\AND\; Q[\beta x])=(P[x]\;\AND\; Q[x])=(P\;\AND\; Q)[x]$
and
$(P\;\OR\; Q)[\beta x]=(P[\beta x]\;\OR\; Q[\beta x])=(P[x]\;\OR\; Q[x])=(P\;\OR\; Q)[x],$
showing that $P\;\AND\; Q$ and $P\;\OR\; Q$ are both left multiplicative.
\end{proof}

As a result of Proposition~\ref{prop:LogicOperators}, the value of~$P[0]$ may be toggled to be either true or false, as desired.

%%%%%%%%%%%%%%%%%%%%%%%%%%%%%%%%%%
\begin{cor}\label{cor:TogglePzero}
Consider the properties $[Q[x]\Longleftrightarrow x=0]$ and $[\hat{Q}[x]\Longleftrightarrow x\neq0]$.
Then~$Q$ and~$\hat{Q}$ are multiplicative.
As a consequence, for any left multiplicative property~$P$ on $R^n$ the properties
$P\;\OR\;Q$ and $P\;\AND\;\hat{Q}$ are again left multiplicative properties.
The former forces $P[0]$ to be true, and the latter forces $P[0]$ to be false.
\end{cor}
%%%%%%%%%%%%%%%%%%%%%%

%%%%%%%%%%%%%%%%%%%%%%%%%%%%%%%%%%%%%%%%%
\section{Orderings of $R$ and $R^n$}\label{SS-Ord}

For the remainder of this paper,~$R$ denotes a (noncommutative) finite principal left ideal ring.
Furthermore, $R^n$ is always considered as a free left $R$-module in the natural way.
We use $R\{v_1,\ldots,v_k\}$ to denote the submodule generated by the vectors $v_1,\ldots,v_k\in R^n$.

For the greedy algorithm in the next section we need a total order on the vectors in~$R^n$.
This will be achieved by picking an ordered basis of~$R^n$ and fixing an order on the scalars in~$R$.
The latter needs to have a specific property for the algorithm to work properly.

%%%%%%%%%%%%%%
\begin{defn}\label{def:Respectful}
A total order~$<$ on~$R$ is called \emph{respectful} if for all $x,y\in R\setminus\{0\}$ it satisfies
\[
     Rx\supsetneq Ry \Longrightarrow \text{ there exists some $\alpha\in R^\ast$ such that $\alpha x<uy$ for all $u\in R^\ast$.}
\]
\end{defn}
%%%%%%%%%%%%%%%%

In combination with Proposition~\ref{prop:AssociateGenerators} this tells us that, in a respectful ordering, for every nonzero
$x\in R$ there is \textit{some} generator of~$Rx$ that comes before \textit{all} nonzero elements of~$Rx$ that are not generators.
The zero element may appear at any position in a respectful order.
Note that any total order of a finite field is respectful, since fields have no proper ideals.

For the existence of respectful orderings on a general finite principal left ideal ring, we need the following concept.
Recall that a \emph{chain} is a totally ordered set.

%%%%%%%%%%%%%%%%%%%%%%%%%%%%%%%%%%%
\begin{defn}\label{def:linearExtension}
A \emph{linear extension} of a partially ordered set $(P,<_P)$ is a chain $(L,<_L)$ equipped with a bijection $f:P\to L$ such that $x <_P y$
implies $f(x) <_L f(y)$.
\end{defn}
%%%%%%%%%%%%%%%%%%%%%%%%%%%%%%
The existence of linear extensions for finite posets is well known~\cite[p.~110]{Sta97}.

\begin{ex}
Consider the poset of ideals of $\Z_6$, ordered by inclusion.
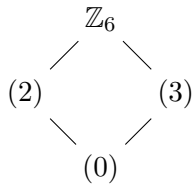
\begin{center}
\begin{tikzpicture}
\node (1) at (0,2) {$\Z_6$};
\node (2) at (-1,1) {$(2)$};
\node (3) at (1,1) {$(3)$};
\node (0) at (0,0) {$(0)$};
\draw (2) -- (1) -- (3) -- (0) -- (2);
\end{tikzpicture}
\end{center}
There are two linear extensions of this poset, namely
\[
  (0)<_L(2)<_L(3)<_L\Z_6\ \text{ and }\ (0)<_L(3)<_L(2)<_L\Z_6.
\]
\end{ex}

%%%%%%%%%%%%%%%%%%%%%%%%%%%%%%%%%%
\begin{defn}
Let $R$ be a principal left ideal ring and~$L$ a linear extension of the poset of left ideals of $R$, ordered by inclusion.
Let~$<$ be a total order on~$R$.
We say that~$<$ \emph{respects} $L$ if $Ry<_L Rx$ implies $x<y$ for all $x,y\in R\setminus\{0\}$.
\end{defn}
%%%%%%%%%%%%%%%%%%%%%%%%%%%%%%%%%%

An ordering that respects a linear extension is indeed respectful, as we shall see in the proof of the next theorem.

%%%%%%%%%%%%%%%%%%%%%%%%%%%%%%%%%%%%%
\begin{theorem}\label{thm:RespectfulExists}
Every finite principal left ideal ring has a respectful ordering.
\end{theorem}
%%%%%%%%%%%%%%%%%%%%%%%%%%%%%%%%%%%%%

\begin{proof}
Let~$R$ be a finite principal left ideal ring and $(L,<_L)$ a linear extension of the poset of left ideals of~$R$.
For each nonzero left ideal~$I$ let~$\Gamma_I$ be the set of its generators.
On each set~$\Gamma_I$ fix an arbitrary total order, denoted by $<_I$.
Then all of this induces an ordering on~$R$ via
\[
   x<y:\Longleftrightarrow \big[Ry<_L Rx\ \text{ or }\ (Ry= Rx=:I\text{ and }x<_I y)\big].
\]
This ordering is respectful and in fact respects $L$.
To see this let $x,y\in R\setminus\{0\}$ such that $Rx\supsetneq Ry$.
Then $Rx\supsetneq Ry=Ruy$ for all $u\in R^*$, hence $Ruy<_L Rx$.
By construction $x<uy$ for all $u\in R^\ast$, which is what we wanted.
\end{proof}

The proof of Theorem~\ref{thm:RespectfulExists} shows that respecting a linear extension is much stronger
than just being respectful.
Instead of showing the existence of \textit{some} unit $\alpha\in R^\ast$ such that $\alpha x<uy$ for all $u\in R^\ast$,
we actually showed that we may pick $\alpha=1$.
This is always the case for orderings that respect a linear extension of the poset of left ideals.
For general respectful orderings, other values of $\alpha$ may be necessary.

%%%%%%%%%%%%%%%%%%%%%%%%%%%%%%%%%%
\begin{ex}\label{E-RespOrder}\
(a)
Consider the ring~$\Z_{12}$ and its poset of ideals
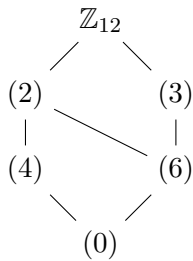
\begin{center}
\begin{tikzpicture}
\node (1) at (0,3) {$\Z_{12}$};
\node (2) at (-1,2) {$(2)$};
\node (3) at (1,2) {$(3)$};
\node (4) at (-1,1) {$(4)$};
\node (6) at (1,1) {$(6)$};
\node (0) at (0,0) {$(0)$};
\draw (2) -- (1) -- (3) -- (6) -- (2) -- (4) -- (0) -- (6);
\end{tikzpicture}
\end{center}
with linear extension $L:\ (0)< (6)<(4)<(3)<(2)<(1)=\Z_{12}$.
Then $1<5<7<11<2<10<3<9<4<8<6<0$ is an ordering of $\Z_{12}$ that respects~$L$.
Note that 1, 5, 7, and 11 generate~$\Z_{12}$; 2 and 10 generate~$(2)$; 3 and 9 generate~$(3)$; 4 and 8 generate~$(4)$;
6 generates~$(6)$.
Recall from Corollary \ref{cor:UnitOrbits} that the generating sets for each ideal are exactly the orbits under
multiplication from $R^\ast$.
So our linear extension induces an order on the set of $R^*$-orbits as well as an order on each $R^\ast$-orbit itself.

(b)
On any integer residue ring~$\Z_m$, the natural order $0<1<\ldots<m-1$ is respectful.
This follows from the fact that the poset of ideals is anti-isomorphic to the poset of positive divisors of~$m$.
However, if~$\Z_m$ is not a field then this order does not respect any linear extension because
$(m-1)=(-1)=\Z_m\supsetneq I$ for any proper ideal~$I$, but
$m-1>a$ for all $a\in\{0,\ldots,m-2\}$  in the natural order.

(c)
The poset of left ideals of $R=M_2(\F_2)$ is given in the diagram below.
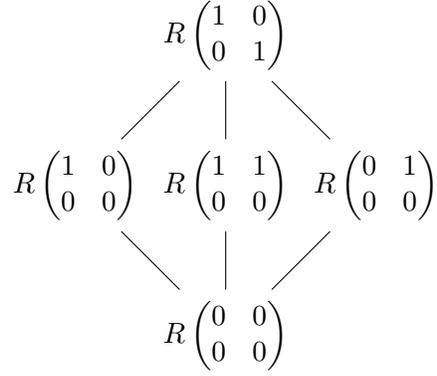
\begin{center}
\begin{tikzpicture}
\node (1) at (0,4) {$R\twoby{1}{0}{0}{1}$};
\node (2) at (-2,2) {$R\twoby{1}{0}{0}{0}$};
\node (3) at (0,2) {$R\twoby{1}{1}{0}{0}$};
\node (4) at (2,2) {$R\twoby{0}{1}{0}{0}$};
\node (0) at (0,0) {$R\twoby{0}{0}{0}{0}$};
\draw (1) -- (2) -- (0) -- (3) -- (1) -- (4) -- (0);
\end{tikzpicture}
\end{center}
We choose a linear extension $L$, which induces the following ordering on the nonzero left $R^\ast$-orbits.
\[
  R^\ast\twoby{0}{1}{0}{0}<_L R^\ast\twoby{1}{0}{0}{0}<_L R^\ast\twoby{1}{1}{0}{0}<_L
  R^\ast\twoby{1}{0}{0}{1}.
\]
Fixing a total order within each $R^*$-orbit, we obtain a respectful ordering on~$R$.
For instance, with the zero matrix as the minimal element we may obtain
\begin{align*}
\twoby{0}{0}{0}{0}
&<\twoby{1}{0}{0}{1}
<\twoby{0}{1}{1}{0}
<\twoby{0}{1}{1}{1}
<\twoby{1}{0}{1}{1}
<\twoby{1}{1}{0}{1}
<\twoby{1}{1}{1}{0}\\
&<\twoby{1}{1}{0}{0}
<\twoby{0}{0}{1}{1}
<\twoby{1}{1}{1}{1}
<\twoby{1}{0}{0}{0}
<\twoby{0}{0}{1}{0}
<\twoby{1}{0}{1}{0}\\
&<\twoby{0}{1}{0}{0}
<\twoby{0}{0}{0}{1}
<\twoby{0}{1}{0}{1}.
\end{align*}

(d)
A total order on  a finite chain ring~$R$ with ideals~\eqref{e-idealChain} is respectful if and only if the following is satisfied:
for any $0\leq i<j\leq e-1$ there is some $\alpha\in R^\ast$ such that $\alpha\gamma^i<u\gamma^j$ for all $u\in R^\ast$.
\end{ex}
%%%%%%%%%%%%%%%%%%%%%%%%%%%%%%%

We now use an ordering on~$R$ to define a lexicographic ordering on~$R^n$.
It is based on a total order of~$R$ together with an ordered basis of~$R^n$.
The total order need not be respectful.
The latter will only be necessary in the next section for the greedy algorithm to produce desirable results.

%%%%%%%%%%%%%%%%%%%%%%%%%%%%%%%%%%%%
\begin{defn}\label{def:LexiOrder}
Let~$R$ be a finite principal left ideal ring with a total order~$<$.
Fix an ordered basis $B=\{b_1,\ldots,b_n\}$ of the free left $R$-module $R^n$.
Let $V_0=\{0\}$ and for $1\leq i\leq n$ let $V_i=R\{b_1,\ldots,b_i\}$ be the submodule of $R^n$ generated by the first~$i$
vectors in $B$.
Thus $V_i=Rb_i+V_{i-1}$.
We define the following \emph{lexicographic ordering} on~$R^n$ and denote it also by $<$:
\begin{arabiclist}
\item If $x\in V_{i-1}$ and $y\in V_i\setminus V_{i-1}$, then set $x<y$.
\item If $x$ and $y$ are distinct vectors in $V_i\setminus V_{i-1}$, then write $x=\sum_{j=1}^i x_j b_j$ and
      $y=\sum_{j=1}^i y_j b_j$, where $x_i,y_i\in R$ are nonzero.
      Let $k\in\{1,\ldots,i\}$ be the highest index such that $x_k\not=y_k$.
      If $x_k<y_k$, then set $x<y$. If $y_k<x_k$, set $y<x$.
\end{arabiclist}
We call $V_i\setminus V_{i-1}$ the \emph{$i$-th level set} of the ordered space~$R^n$.
\end{defn}
%%%%%%%%%%%%%%%%%%%%%%%%%%%%%%%%%%%%%

The ordering defined above can be described in easier terms as follows.
Denote by $<_{\rm lex}$ the lexicographic ordering on~$R^n$ induced by the respectful ordering~$<$ on~$R$.
Then for $x\in V_l\setminus V_{l-1}$ and $y\in V_m\setminus V_{m-1}$ we have
\begin{equation}\label{e-SortLevel}
  x<y\Longleftrightarrow \big[ l<m\ \text{ or }\ (l=m \text{ and }(x_n,\ldots,x_1)<_{\rm lex}(y_n,\ldots,y_1))\big],
\end{equation}
where $(x_1,\ldots,x_n),\,(y_1,\ldots,y_n)$ are the coefficient
vectors of~$x$ and~$y$  with respect to the chosen basis~$B$, that is $x=\sum_{j=1}^n x_j b_j$ and $y=\sum_{j=1}^n y_j b_j$.
In the case where~$0$ is the least element of the ordered ring~$R$, this even simplifies to
\[
  x<y\Longleftrightarrow (x_n,\ldots,x_1)<_{\rm lex}(y_n,\ldots,y_1).
\]
Yet in other terms, we have the ordering of levels
\begin{equation}\label{e-SortWithinLevel}
  \{0\}=V_0\;<\;V_1\setminus V_0\;<\;V_2\setminus V_1\;<\,\ldots\,<\; V_n\setminus V_{n-1},
\end{equation}
where each level set is ordered according to~\eqref{e-SortLevel}.
Thus the ordering on~$R$ only dictates the ordering within each level set, but not the ordering between the levels.
The latter is dictated by the chosen ordered basis~$B$.

%%%%%%%%%%%%%%%%%%%%%%%%%%%%%%%%%
\begin{ex}\label{ex:Z4NaturalOrder}\
(a)
Let $\Z_4$ be endowed with the natural ordering $0<1<2<3$.
Let $B=\{001,010,100\}$ be the reverse standard basis for $\Z_4^3$.
Then Definition \ref{def:LexiOrder} leads to the natural lexicographic ordering
\[
 000<001<002<003<010<011<012<013<020<021<022<023<\ldots<333.
\]
This is simply the lexicographic ordering because for the reverse standard basis the reversed coefficient vector of~$x$ (see~\eqref{e-SortLevel})
is simply~$x$ itself.

(b)
Let now $\Z_4$ be equipped with the ordering $1<3<2<0$, which respects the chain of ideals $(1)\supsetneq(2)\supsetneq(0)$.
Let $B=\{100,010,001\}$ be the standard basis for $\Z_4^3$.
Then the lexicographic ordering from Definition \ref{def:LexiOrder} is given by~\eqref{e-SortWithinLevel} and the internal ordering:
\begin{align*}
  V_0\quad &=\{000\},\\
  V_1\setminus V_0&=\{100<300<200\},\\
  V_2\setminus V_1&=\{110<310<210<010<130<330<230<030<120<320<220<020\},\\
  V_3\setminus V_2&=\{111<311<211<011<131<\ldots<002\}.
\end{align*}
Notice that the zero element acts here in two different ways: it ``naturally'' sorts the levels $V_i\setminus V_{i-1}$,
but dictates an unusual sorting within each level.
\end{ex}
%%%%%%%%%%%%%%%%%%%%%%%%%%%%%%%%%%%%%%%%%%%%

%%%%%%%%%%%%%%%%%%%%%%%%%%%%%%%%%%%%%%%%%%
\section{The Greedy Algorithm}\label{SS-Alg}
We now introduce a greedy algorithm that produces codes over a given finite principal left ideal ring such that all (nonzero) codewords have a given
pre-specified property.
The algorithm generalizes the ones presented by Van Zanten and Suparta in~\cite{VZS05} for codes over finite fields and by Guenda et al.
in~\cite{GGS14} for codes over finite chain ring.

Throughout, let~$R$ be a finite principal left ideal ring.
Moreover, let~$\Gamma$ be a fixed set of generators of the nonzero left ideals in~$R$.
The following algorithm itself does not need the respectfulness of the ordering on~$R$, but the properties of the resulting codes heavily rely on it.
Thus we restrict ourselves to respectful orderings on~$R$.

%%%%%%%%%%%%%%%%
\begin{algorithm}\label{alg:Greedy}
Fix a respectful ordering $<$ on~$R$ and an ordered basis $B$ of~$R^n$.
Consider the resulting lexicographic ordering on the left $R$-module $R^n$ as in Definition \ref{def:LexiOrder}.
Let $P$ be a left multiplicative property on $R^n$.
\begin{arabiclist}
\item Put $C_0=\{0\}$. Set $i=1$.
\item Search for the first (smallest) vector $a_i\in V_i\setminus V_{i-1}$ such that $P[\gamma a_i+c]$ holds true for all $\gamma\in\Gamma$ and all $c\in C_{i-1}$.
\item $\sbt$ If such $a_i$ exists, let $C_i:=\{r a_i+c\mid r\in R,\,c\in C_{i-1}\}=Ra_i+C_{i-1}$.
      \\
      $\sbt$ If no such $a_i$ exists, let $C_i:=C_{i-1}$.
\item $\sbt$ If $i<n$, set $i:=i+1$ and return to Step~2.
      \\
      $\sbt$ If $i=n$, stop and output $C_n$.
\end{arabiclist}
We call $C_n$ a \emph{lexicode} (or \emph{lexicographic code}) with respect to the given ordering, basis, and property~$P$ and denote it by $C(<,B,P)$.
\end{algorithm}
%%%%%%%%%%%%%%%%%

The generated codes~$C_i$ clearly depend on the chosen basis~$B$, which determines the sets~$V_i$ and thus the level sets~$V_i\setminus V_{i-1}$,
as well as on the ordering on~$R$, which determines the ordering within the level sets.
Examples of this dependence will be provided in Section~\ref{SS-Exa}.

We wish to point out that we explicitly allow multiplicative properties~$P$ for which $P[0]$ is false.
While this may seem odd because we aim at constructing linear codes, this does indeed lead to interesting outcomes --
as we will show later.
Note that the algorithm always adds the zero vector to the code.

The following lemma shows that the selection criterion $P[\gamma a_i+c]$ for all $\gamma\in\Gamma,\,c\in C_{i-1}$ is sufficient to actually
guarantee $P[ra_i+c]$ for all $r\in R\,\backslash\,\{0\},\,c\in C_{i-1}$.
This implies that the resulting sets~$C_i$ do not depend on the choice of the generator set~$\Gamma$.
The use of~$\Gamma$ in the algorithm merely serves to reduce the number of tests in the selection step (Step~2.).
If~$R$ is a finite field, then we may choose $\Gamma=\{1\}$, and the algorithm reduces to Algorithm~A in
\cite{VZS05} by Van Zanten and Suparta.
If~$R$ is a finite chain ring with ideals as in~\eqref{e-idealChain} we may choose $\Gamma=\{1,\gamma,\gamma^2,\ldots,\gamma^{e-1}\}$,
and the algorithm equals Algorithm~A in \cite{GGS14} by Guenda et al.
As the proof of the following lemma shows, the multiplicativity of the property~$P$ is crucial.
See also Example~\ref{ex:counterex}(a) for a trivial counterexample showing how the lemma fails when~$P$ is not multiplicative.

%%%%%%%%%%%%%%%%%%%%%%%%%%%
\begin{lemma}\label{lem:GammaSufficient}
Let $P$ be a left multiplicative property on $R^n$, and $C$ a left submodule of $R^n$ such that $P[c]$ holds for all nonzero $c\in C$.
Let $x\in R^n$. Then
\[
    P[\gamma x+c]\text{ holds true for all }\gamma\in\Gamma,\,c\in C\Longleftrightarrow P[rx+c]\text{ holds true for all }r\in R\setminus\{0\},\,c\in C.
\]
\end{lemma}
%%%%%%%%%%%%%%%%%%%%%%%%%%%

\begin{proof}
We only need to prove ``$\Longrightarrow$''.
For $r\in R\setminus\{0\}$, we have $Rr=R\gamma$ for some $\gamma\in\Gamma$.
Hence Proposition \ref{prop:AssociateGenerators} yields $r=u\gamma$ for some $u\in R^\ast$.
Then $P[rx+c]=P[u\gamma x+c]=P[\gamma x+u^{-1}c]$ for all $c\in C$ by multiplicativity of~$P$.
Since~$C$ is linear, the latter is true for all $c\in C$ by assumption, and the proof is complete.
\end{proof}

The next theorem generalizes \cite[Theorem 2.2]{VZ97} for lexicodes over $\F_2$, \cite[Theorem 2.2]{VZS05} for lexicodes over $\F_q$, and
\cite[Theorem 4]{GGS14} for lexicodes over finite chain rings.

%%%%%%%%%%%%%%%%%%%%%%%%%%%%%%
\begin{theorem}\label{thm:PropertyHolds}
Consider Algorithm \ref{alg:Greedy}.
Then each set $C_i$ is a code, i.e., a submodule of $R^n$, and $P[x]$ is true for every nonzero codeword $x\in C_i$.
\end{theorem}
%%%%%%%%%%%%%%%%%%%%%%%%%%%%%%%%%%

\begin{proof}
Left linearity of~$C_i$ is clear. Vacuously $P[x]$ holds for all nonzero $x\in C_0$.
Suppose now that $P[x]$ holds for all nonzero $x\in C_{i-1}$.
If $C_i=C_{i-1}$, then there is nothing to prove. Else let~$a_i$ be the selected vector from $V_i\setminus V_{i-1}$.
Then by Lemma~\ref{lem:GammaSufficient} $P[ra_i+c]$ holds true for all $r\in R\setminus \{0\}$ and all $c\in C_{i-1}$.
Since $C_i=Ra_i+C_{i-1}$, this establishes the desired result.
\end{proof}

Note that in Step~2 of Algorithm~\ref{alg:Greedy} we only select one (if any) vector~$a_i$ in the level $V_i\setminus V_{i-1}$, update
$C_{i-1}$ to $C_i:=Ra_i+C_{i-1}$,  and then move on to the next level $V_{i+1}\setminus V_i$.
The next theorem justifies abandoning the search through the rest of $V_i\setminus V_{i-1}$.
Indeed, as we will see, the respectfulness of the ordering on~$R$ guarantees that any vector $x\in V_i\setminus V_{i-1}$ such that $P[\gamma x+c]$ is true
for all $\gamma\in\Gamma$ and $c\in C_i$ is already in $C_i$.
Therefore this theorem generalizes the result of \cite[Theorem 2.1]{VZS05} for lexicodes over $\F_q$, and the result of
\cite[Lemma 3]{GGS14} for lexicodes over finite chain rings.

%%%%%%%%%%%%%%%%%%%%%%%%%%%%%%%%%
\begin{theorem}\label{thm:SkipCheck}
Consider Algorithm \ref{alg:Greedy} and the resulting nested codes $C_0\subseteq\ldots\subseteq C_n$.
Then every vector $x\in V_i\setminus V_{i-1}$ satisfying $P[\gamma x+c]=\text{true}$ for all $\gamma\in\Gamma$ and all $c\in C_i$, is already in $C_i$.
\end{theorem}
%%%%%%%%%%%%%%%%%%%%%%%%%%%%%%%%%%%

\begin{proof}
We induct on $i$. For the base case, the statement is trivially true because $V_0=\{0\}=C_0$.

For $1\leq i\leq n$, assume the statement holds for all indices less than $i$. Suppose $x\in V_i\setminus V_{i-1}$ is such that
$P[\gamma x+c]$ holds true for all $\gamma\in\Gamma$ and all $c\in C_i$.
Then there must have been some selected vector $a_i\in V_i\setminus V_{i-1}$ such that $P[\gamma a_i+c]$ holds true for all
$\gamma\in\Gamma$ and $c\in C_{i-1}$.
Thus $C_i=Ra_i+C_{i-1}$.

Write $a_i=p_ib_i+\sum_{l=1}^{i-1}p_lb_l$ and $x=q_ib_i+\sum_{l=1}^{i-1}q_lb_l$, where $p_l,q_l\in R$ and $p_i\neq0\neq q_i$.
Then $Rp_i+Rq_i=Rd$ for some $d\in R$, since $R$ is a principal  left ideal ring.
Hence there exist $a,b\in R$ such that $ap_i+bq_i=d$, and by
Theorem~\ref{thm:UnitLinearComb}, we may even assume that $b$ is a unit. Let $y=a a_i+bx$.
Then for every $u\in R^\ast$ we have $\gamma ub\not=0$ because $ub$ is not a zero divisor.
Hence by our assumption on $x$ and Lemma~\ref{lem:GammaSufficient}
\begin{equation}\label{e-Pguy}
    P[\gamma u y+c]=P[\gamma ubx+(\gamma ua a_i+c)]\ \text{ holds true for all }\gamma\in\Gamma,\,u\in R^*,\,c\in C_i,
\end{equation}
since $\gamma ua a_i+c\in C_i$.
Now observe that
\[
  y=a a_i+bx
  =\Big(a p_ib_i+a\sum_{l=1}^{i-1}p_lb_l\Big)+\Big(bq_ib_i+b\sum_{l=1}^{i-1}q_lb_l\Big)=db_i+\sum_{l=1}^{i-1}(ap_l+bq_l)b_l.
\]
By construction, $Rd \supseteq Rp_i$. If $Rd\supsetneq Rp_i$, then our respectful ordering dictates that there is some
$\alpha\in R^\ast$ such that $\alpha d<p_i$.
Thus $\alpha y<a_i$.
But $P[\gamma\alpha y+c]$ holds true for all $\gamma\in\Gamma,\,c\in C_i$ by~\eqref{e-Pguy},
so $\alpha y$ would have been selected instead of $a_i$, a contradiction.
Hence we must have $Rd =Rp_i$, and thus $Rq_i\subseteq Rd=Rp_i$.
So, there exists some $\beta\in R\setminus\{0\}$ such that $\beta p_i=q_i$.
Then
\begin{equation}\label{e-xRel}
    x=\beta a_i+v\ \text{ for some }\ v\in V_{i-1}.
\end{equation}
Now we are ready to show that $x\in C_i$. We will do so by proving that $v\in C_{i-1}$.
Let $\gamma\in\Gamma$ and $c'\in C_{i-1}$.
Define $c:=-\gamma\beta a_i+c'$. Then $c\in C_i$ and
\[
  P[\gamma v+c']=P[\gamma(x-\beta a_i)+c+\gamma\beta a_i]=P[\gamma x+c],
\]
hence $P[\gamma v+c']$ holds true by assumption on~$x$.
Now our induction hypothesis implies that $v\in C_{i-1}$ and thus $x=\beta a_i+v$ is in $C_i$, as desired.
\end{proof}

In the proof of Theorem \ref{thm:SkipCheck}, we introduced the vector~$y$ for the sole purpose of showing that the ideals~$Rp_i$
and~$Rq_i$ are comparable in the poset of left ideals.
For the case of finite chain rings, all left ideals are comparable
and the containment $Rq_i\subseteq Rp_i$ follows immediately from the
respectful ordering, so the proof becomes greatly simplified.

As the proof above suggests, the existence of $\beta\in R\setminus\{0\}$ and $v\in V_{i-1}$ such that $x=\beta a_i+v$ is not
trivial over rings (it is clearly always the case over fields).
Only the respectfulness of the ordering on~$R$ guarantees this step for principal left ideal rings, and
in Example~\ref{E-Z4isotropic}(b) we show that the theorem above is indeed not true if the ordering of~$R$
is not respectful.
For this reason our proof completes the one given in \cite[Lemma~3]{GGS14},
where this detail seems to have been overlooked since no specifics on the ordering of the ring elements are given.
It  seems, however, that only respectful orderings were used in the examples in~\cite{GGS14}.

In Example~\ref{ex:counterex} we show that the previous theorem also fails if either the property~$P$ is not left multiplicative
or the ring is not a principal left ideal ring.

The examples in the next section suggest that the use of a respectful ordering in Algorithm \ref{alg:Greedy} produces large codes.
As we show next, these codes are in fact maximal if $P[0]$ is true.
The maximality in the sense of the following theorem is not true if $P[0]$ is false; see Example~\ref{ex:Z4Lee}.
But we do obtain a certain analogy for the case where $P[0]$ is false, as we will show below.
Recall from Corollary~\ref{cor:TogglePzero} that we may toggle the value of $P[0]$ as desired.
For instance, we may overwrite the value for the familiar property $[P[x]\Longleftrightarrow \wtH(x)\geq\delta]$ and toggle $P[0]$ to true.

%%%%%%%%%%%%%%%%%%%%%%%%%%%%%%%%
\begin{theorem}\label{thm:TrueMaximal}
Let~$<$ and~$B$ be as in Algorithm \ref{alg:Greedy} and let~$P$ be a left multiplicative property such that~$P[0]$ is true.
Then the lexicode $C(<,B,P)$ is maximal (with respect to inclusion) in the poset of all codes satisfying~$P$.
\end{theorem}
%%%%%%%%%%%%%%%%%%%%%%%%%%%%%%%%%%%%

\begin{proof}
Recall the codes~$C_i$ from Algorithm~\ref{alg:Greedy}.
Suppose contrarily that there is some linear code $C$ satisfying $P$ such that $C_n\subsetneq C\subseteq R^n$.
Let $x\in C\setminus C_n$.
Then $P[\gamma x+c]$ holds for all $\gamma\in\Gamma,\,c\in C_n$ by assumption.
Since $x$ lies in $V_i\setminus V_{i-1}$ for some $i=1,\ldots,n$ and $P[\gamma x+c]$ holds for all $\gamma\in\Gamma,\,c\in C_i$ (even if $\gamma x+c=0$),
Theorem \ref{thm:SkipCheck} implies that $x\in C_i\subseteq C_n$, a contradiction.
\end{proof}

We now turn to the case where $P[0]$ is false.

%%%%%%%%%%%%%%%%
\begin{theorem}\label{thm:FalseFree}
Let~$<$ and~$B$ be as in Algorithm \ref{alg:Greedy} and let~$P$ be a left multiplicative property such that $P[0]$ is  false.
Then each code~$C_i$ generated by Algorithm \ref{alg:Greedy} is free, and the selected vectors form a basis for $C_i$.
\end{theorem}
%%%%%%%%%%%%%%%

\begin{proof}
Let $a_{j_1}<\ldots<a_{j_k}$ be the vectors selected by Algorithm \ref{alg:Greedy} to generate $C_i$.
Suppose that the vectors are linearly dependent, say
$\sum_{l=1}^{t}\lambda_l a_{j_l}=0$
for some scalars $\lambda_l\in R$ with $\lambda_t\not=0$.
Note that $a_{j_1}<\ldots<a_{j_t}$ generate some $C_{i'}$ and $a_{j_1}<\ldots<a_{j_{t-1}}$ are in $C_{i'-1}$.
Lemma \ref{lem:GammaSufficient} tells us that $P[ra_{j_t}+c]$ holds true for every $r\in R\setminus\{0\}$ and $c\in C_{i'-1}$.
In particular $P[\lambda_ta_{j_t}+\sum_{l=1}^{t-1}\lambda_l a_{j_l}]$ is true, contradicting that $P[0]$ is false.
Therefore the vectors $a_{j_1},\ldots,a_{j_k}$ form a linearly independent set.
Since by construction~$C_i$ is generated by these vectors, we obtain the desired result.
\end{proof}

We now obtain the analogue of Theorem~\ref{thm:TrueMaximal}.
%%%%%%%%%%%%%%%%%%%%%%%%%%%
\begin{theorem}\label{thm:FalseMaxFree}
If~$<$ is a respectful ordering of~$R$ and~$P$ is a left multiplicative property where $P[0]$ is false,
then the code $C:=C(<,B,P)$ generated by Algorithm~\ref{alg:Greedy} is maximal  (with respect to inclusion) in the poset of all free
codes satisfying~$P[x]$ for all nonzero $x\in C$.
\end{theorem}
%%%%%%%%%%%%%%%%%%%%%%%%%%%%%%%%%%

\begin{proof}
By Theorems \ref{thm:FalseFree} and \ref{thm:PropertyHolds}, the module~$C$ is free with basis $\{a_{j_1},\ldots,a_{j_k}\}$,
and all nonzero codewords in~$C$ satisfy~$P$.
Suppose contrarily that there is some free linear code~$\tilde{C}$ with all nonzero codewords satisfying~$P$ and such that
$C\subsetneq\tilde{C}\subseteq R^n$.
By Theorem~\ref{T-FLPIR} there exists a submodule~$C'$ such that $C\oplus C'=\tilde{C}$.
Hence there exists some $x\in\tilde{C}\setminus C$ such that $\{a_{j_1},\ldots,a_{j_k},x\}$ is linearly independent.
Thus $r x+c\not=0$ for all $r\in R\setminus\{0\},\,c\in C$ and therefore $P[r x+c]$ is true for all these vectors.
Let $i\in\{1,\ldots,n\}$ such that $x\in V_i\setminus V_{i-1}$.
But then Theorem~\ref{thm:SkipCheck} tells us that $x\in C_i$, contradicting that $x\not\in C$.
\end{proof}

In Examples~\ref{ex:Z4Lee} and~\ref{ex:Z10homog} we illustrate the different outcomes of the algorithm when we toggle $P[0]$
between true and false.
In general, but not always, if~$P[0]$ is false one obtains a significantly smaller code.
More importantly, even though toggling $P[0]$ to true simply widens the selection criterion, the algorithm does not always
produce a lexicode that contains the lexicode for $P[0]$ being false.

The following result shows that with a suitable choice of the lexicographic ordering on~$R^n$ every free code satisfying some multiplicative property
can be obtained as a `partial lexicode`, that is, a code obtained when stopping the algorithm after a certain number of rounds.
In combination with the previous theorems this result may be used to test whether a given code is maximal among all codes satisfying the property or, if not, extend it to a maximal code.

%%%%%%%%%%%%%%%%%%%%%%%%%%%%%%%%%%%%%
\begin{theorem}\label{thm:AllCodesLexi}
Any free linear code $C\subseteq R^n$ satisfying some multiplicative property $P$ for all nonzero $x\in C$
is a subcode of a lexicode $C(<,B,P)$ for some suitable respectful ordering~$<$ on~$R$ and a suitable basis~$B$ of~$R^n$.
\end{theorem}
%%%%%%%%%%%%%%%%%%%%%%%%%%%%%%%%%%%

\begin{proof}
Since $C$ is free, it has some basis $\{b_1,\ldots,b_k\}$.
By Theorem~\ref{T-FLPIR}, we can extend this to a basis $B=\{b_1,\ldots,b_n\}$ of $R^n$.
Choose a respectful ordering on~$R$ starting with $0<1<\ldots\;$.
Running Algorithm~\ref{alg:Greedy} with this basis~$B$ and respectful ordering, the first vector in the level set $V_i\setminus V_{i-1}$
is~$b_i$.
Note that $r b_i+c\not=0$ for all $r\in R\setminus\{0\},\,c\in V_{i-1}$.
Thus, for $i\leq k$ we have that $P[r b_i+c]$ holds true (regardless of the value of $P[0]$),
and thus the algorithm selects $a_i=b_i$ for every $i=1,\ldots, k$.
Then $C=C_k\subseteq C(<,B,P)$.
\end{proof}

Based on an abundance of examples, we strongly believe that Theorem~\ref{thm:AllCodesLexi} is true for general (i.e., non-free) codes.
Unfortunately, we are not able to provide a proof at this point.

%%%%%%%%%%%%%%%%%%
\section{Examples of Lexicodes}\label{SS-Exa}
We start with an example showing that the respectfulness of the ordering is necessary for Theorem~\ref{thm:SkipCheck} to be true, even over a finite chain ring.

%%%%%%%%%%%%%%%%%%%%%%%%%%%%%%%
\begin{ex}\label{E-Z4isotropic}
(a) Consider $\Z_4^4$ with the standard basis $B=\{1000,0100,0010,0001\}$ and the property $[P[x]\Longleftrightarrow x\cdot x=0]$.
Using the natural, thus respectful, ordering $0<1<2<3$, the selected vectors are $a_1=2000, a_2=0200, a_3=0020$, and $a_4=1111$.
The resulting lexicode $C_4=R\{a_1,a_2,a_3,a_4\}$
has cardinality~$32$. It is not free (because its cardinality is not a power of~$4$).

(b)
Consider now the non-respectful ordering $0<2<1<3$ on~$\Z_4$.
Using the same basis of~$\Z_4^4$ and the same property as in~(a),
the algorithm generates the lexicode~$C_4=R\{a_1,a_2,a_3,a_4\}$ of size 16 with selected vectors $a_1=2000, a_2=0200, a_3=0020$, and $a_4=0002$.
Observe that this code is strictly contained in the one from~(a).
The vector $x=1111\in V_4\setminus V_3$ satisfies $P[\gamma x+c]$ for all $\gamma\in \Gamma=\{1,2\}$ and $c\in C_3$ (the code from the previous iteration of the algorithm).
But $x$ is not in~$C_4$.
This is due to the fact that~$x$ cannot be written in the form $x=\beta a_4+v$ for
any $\beta\in R,\,v\in V_3$; see~\eqref{e-xRel} in the proof of Theorem~\ref{thm:SkipCheck}.
In other words, the vector 0002 was selected instead of 1111 (or some other vector with a unit in the last entry), which would not have happened with a respectful ordering.
\end{ex}
%%%%%%%%%%%%%%%%%%%%%%%%%%%%%%%

The following examples show that Theorem~\ref{thm:SkipCheck} is not true in general if either the property is not left
multiplicative or the ring is not a principal left ideal ring.

%%%%%%%%%%%%%%%%%%%%%%%%%%%%%%%%%%%%%%%%%
\begin{ex}\label{ex:counterex}
(a) Consider the field~$R=\Z_3=\{0,1,2\}$ with the natural order $0<1<2$, which is respectful.
Let~$P$ be the property $[P[x]\Longleftrightarrow x=2]$.
Then~$P$ is not multiplicative because $P[2\!\cdot\!2]\neq P[2]$.
In~$R^1$ with standard basis~$e_1=1$ the lexicode resulting from  Algorithm~\ref{alg:Greedy} is $C=\Z_3$.
It does not satisfy Theorem~\ref{thm:PropertyHolds}.
Note that due to the non-multiplicativity of~$P$ even Lemma~\ref{lem:GammaSufficient} is not true.

(b) Consider the ring $R:=\F_2[x,y]/(x^2,xy,y^2)=\{0,x,y,x+y,1,1+x,1+y,1+x+y\}$. Note that the last 4 elements are the units of~$R$.
The ring has 4 non-trivial ideals given by $(x),\,(y),\,(x+y),\,(x,y)$. The first three are principal and have cardinality~$2$,
the last one is not principal and has cardinality~$4$.
Based on this and Definition~\ref{D-HomogWeight} the homogeneous weight on~$R$ turns out to be
\[
  \omega(0)=0,\quad \omega(x)=\omega(y)=\omega(x+y)=2,\quad \omega(u)=\frac{1}{2}\text{ for all }u\in R^*.
\]
In~$R^1$ consider the multiplicative property $[P[x]\Longleftrightarrow \omega(x)\geq2\text{ or }x=0]$.
Moreover, consider any ordering of the ring elements\footnote{Note that the definition of respectfulness for an ordering is based
on principal  left ideals. If we simply ignore the non-principal ideal $(x,y)$ and follow Definition~\ref{def:Respectful}, then any ordering of the form
$0<1<\text{``rest''}$ may be called respectful.} and the standard basis $e_1=1$.
Then Algorithm~\ref{alg:Greedy} results in the lexicode $C=C_1=(w)=\{0,w\}$, where~$w$ is the first nonunit element in the ordering of~$R$.
As a consequence, Theorem~\ref{thm:SkipCheck} is not satisfied for $i=1$ because every element in $\{0,x,y,x+y\}$
satisfies the property.
For the same reason, Theorem~\ref{thm:TrueMaximal} is not satisfied.
All of this shows that for non-principal ideal rings, the search in Step~2.\ of Algorithm~\ref{alg:Greedy} should continue through each entire level $V_i\setminus V_{i-1}$.
\end{ex}
%%%%%%%%%%%%%%%%%%%%%%%%%%%%%%%%%%%%%

The next example illustrates that different respectful orderings may generate different codes.
Part~(b) shows that, for codes over fields, even the dimension of the resulting code
depends on the choice of the respectful ordering.

%%%%%%%%%%%%%%%%%%%%%%%%%%%%%%%%%%%
\begin{ex}\label{ex:DifferentOrderings}
(a)  Consider the reverse standard basis $B=\{001,010,100\}$ for $\Z_4^3$ and the selection property
       $[P[x]\Longleftrightarrow\wtL(x)\geq 2]$, where $\wtL$ is the Lee weight; see
       Definition~\ref{D-Weights}(d). Note that $P[0]$ is false.
       Since  $\Z_4^*=\{1,3\}$, a total ordering~$<$ on~$\Z_4$ is respectful iff $1<2$ or $3<2$.
       We obtain the following cases:
       \\
       (i) Using any of the respectful orderings $r_1<0<r_2<r_3$, where $r_1\in\Z_4^*$, we obtain the lexicode
        $C=\Z_4\{011,103\}$ (the two given vectors are not necessarily the vectors $a_i$ selected by the algorithm).
        \\
      (ii) With any of the respectful orderings $r_1<r_2<r_3<r_4$, where $\{r_1,r_2\}=\Z_4^*$, we obtain the lexicode
       $C=\Z_4\{011,102\}$.
       \\
       (iii)  With any other respectful ordering we obtain the lexicode $C=\Z_4\{011,101\}$.
       \\
       In each case the resulting lexicode has cardinality~$16$.

(b)
 Consider the field~$\F=\F_7$ and in~$\F^3$ define the codes $C=\F\{100,010\}$, $D=\F\{001\}$.
Let~$P$ be the property $[P[x]\Longleftrightarrow x\in C\cup D]$.
Note that~$P$ is  multiplicative and $P[0]$ is true.
Fix the ordered basis $B=\{113,331,100\}$ of~$\F^3$.
\\
i) Using the respectful ordering $0<1<2<3<4<5<6$ the algorithm returns $a_2=1(331)+2(113)=550$, thus $C_2=\F\{550\}$, and
$a_3=100$. Hence $C(<,B,P)=C$.
\\
ii) Using the respectful ordering $0<1<4<3<2<5<6$ the algorithm returns $a_2=1(331)+4(113)=006$, thus
$C_2=\F\{001\}$, and there is no return for $a_3$.
Thus $C(<,B,P)=D$.
\end{ex}
%%%%%%%%%%%%%%%%%%%%%%%%%%%%%%%%%%%%%%%

Of course, the output of the algorithm also depends on the choice of the basis~$B$.
Again, even the dimension of the lexicode (e.g., for field alphabets) depends on~$B$.
The choice of basis may also decide on whether the lexicode is free or not.
%%%%%%%%%%%%%%%%%%%%%%%%%%%%%%%%%
\begin{ex}\label{E-FieldDim}
(a)
Let~$\F$ be any finite field and in~$\F^3$ consider the two codes $C=\F\{100,010\}$ and $D=\F\{001\}$.
Let~$P$ be the property $[P[x]\Longleftrightarrow x\in C\cup D]$.
Fix any total ordering~$<$ on~$\F$.
Using the basis $B=\{100,010,001\}$, the greedy algorithm returns the code~$C(<,B,P)=C$, whereas with the basis
$B'=\{001,010,100\}$ it returns $C(<,B',P)=D$.

(b) Consider the codes $C=\Z_4\{200,020\},\,D=\Z_4\{001\}$ in the module $\Z_4^3$.
Using the same property as in~(a) and the standard basis of~$\Z_4^3$, the algorithm returns the non-free code~$C$, whereas with the reverse standard basis it returns the free code~$D$.
\end{ex}
%%%%%%%%%%%%%%%%%%%%%%%%%%%%%%%

We now illustrate the outcome of the greedy algorithm when toggling $P[0]$ between true and false.
%%%%%%%%%%%%%%%%%%%%%%%%%%%%%%%
\begin{ex}\label{ex:Z4Lee}
Consider any respectful ordering on~$\Z_4$ and the module $\Z_4^3$ with the standard basis.
Let~$P$ be the multiplicative property $[P[x]\iff \wtL(x)\geq 6]$, where $\wtL$ is again the Lee weight.

(a) The only vector in $\Z_4^3$ that satisfies~$P$ is $222$.
      But since $2\cdot 222=000$ and $P[000]$ is false, Algorithm \ref{alg:Greedy} returns the zero code.

(b) If we toggle $P[000]$ to true, then $222$ is selected and we obtain the non-free code $\{000,222\}$.
\end{ex}
%%%%%%%%%%%%%%%%%%%%%%%%%%%%%%%%%

%%%%%%%%%%%%%%%%%%%%%%%%
\begin{ex}\label{ex:Z4dot}
We consider the exact situation of Example~\ref{E-Z4isotropic}(a) with the only difference that we toggle $P[0]$ to false.
Thus the property is $[P[x]\Longleftrightarrow x\cdot x=0 \text{ and }x\neq0]$.
Using the same respectful ordering and the same basis, Algorithm~\ref{alg:Greedy} returns the code $C=\Z_4\{1111\}$.
It is a free subcode of the lexicode returned in Example~\ref{E-Z4isotropic}(a).
\end{ex}
%%%%%%%%%%%%%%%%%%%%%%%%

%%%%%%%%%%%%%%%%%%%%%%%%%%%%%%%%%%%
\begin{ex}\label{ex:Z10homog}
Consider $R:=\Z_{10}$ with the natural, thus respectful, ordering $0<1<\ldots<9$.
By Definition~\ref{D-HomogWeight} the homogeneous weight on~$R$ is given by
\begin{center}
\begin{tabular}{|c|cccccccccc|}
\hline
$x$ & 0 & 1 & 2 & 3 & 4 & 5 & 6 & 7 & 8 & 9 \\
\hline
$\omega(x)$ & $0$ & $\frac{3}{4}$ & $\frac{5}{4}$ & $\frac{3}{4}$ & $\frac{5}{4}$ & $2$ & $\frac{5}{4}$ & $\frac{3}{4}$ & $\frac{5}{4}$ & $\frac{3}{4}$ \\
\hline
\end{tabular}
\end{center}

(a) Consider now the multiplicative property $[P[x]\iff \omega(x)\geq 2]$ on the module $R^3$, where the homogeneous weight is extended
      to vectors as in~\eqref{e-ExtWeight}.
      Thus $P[0]$ is false.
      Using the ordered basis $B=\{001,010,100\}$ of $R^3$, Algorithm~\ref{alg:Greedy} returns
      $C_1=\{0\},\,C_2=C_3=R\{012\}$.
      Hence $C:=C(<,B,P)=C_3$ is indeed a free code with basis $\{012\}$ and cardinality~$10$.

(b) With the same data as in~(a), but where we toggle $P[0]$ to true, the algorithm returns %$C_1=R(005),\,C_2=R\{(005),(021)\}$, and
      $C':=C_3=R\{005,021,201\}$.
      The code~$C'$ is not free and has cardinality~$50$.
      A minimal generating set is given by $\{201,820\}$.
      One should note that the code~$C$ from~(a), which is free, is not a subcode of~$C'$.
      In fact, $C\cap C'=\{0\}$, though all we did is toggle $P[0]$.
\end{ex}
%%%%%%%%%%%%%%%%%%%%%%%%%%%%%%%%%%%%%

Let us now turn to the construction of self-orthogonal codes.
Recall from Remark~\ref{rem:Isotr} that over a commutative ring with odd characteristic we achieve
self-orthogonality using the (multiplicative) property $[P[x]\Longleftrightarrow x\cdot x=0]$.
Obviously $P[0]$ is true.
In~(c) of the following example we provide a case where overwriting $P[0]$ to false produces a free code \emph{of the same size}
as the lexicode for the case where $P[0]=$ true.
The fact that we obtain in both cases ($P[0]$ true or false) lexicodes of the same size is remarkable because, more often than not,
codes generated with $P[0]=$ false are much smaller than their counterparts with $P[0]=$ true.

%%%%%%%%%%%%%%%%%%%%%%%%%%%%%%%
\begin{ex}
For all examples we consider the property $\big[P[x]\Longleftrightarrow x\cdot x=0\big]$.

(a) On $\F_5^4$ consider the reverse standard basis $B=\{e_4,e_3,e_2,e_1\}$  and
      fix the natural ordering $0<1<2<3<4$ on~$\F_5$.
      Then Algorithm~\ref{alg:Greedy} returns $C_1=\{0\},C_2=C_3=\F_5\{0012\}$ and $C(<,B,P)=C_4=\F_5\{0012,1200\}$,
      which by Theorem~\ref{thm:PropertyHolds} and Remark~\ref{rem:Isotr} is self-orthogonal, that is, $C\subseteq C^\perp$.
      Using $\dim(C)+\dim(C^{\perp})=n$ and $\dim(C)=2$ we conclude that $C=C^\perp$, that
      is,~$C$ is self-dual. This also shows that~$C$ is not a proper subcode of a code satisfying property~$P$ (thus illustrating
      Theorem~\ref{thm:TrueMaximal}).

(b) In the same way we obtain in $\F_7^4$ (using the natural ordering and the reverse standard basis)
      the self-dual code $C=\F_7\{0123,1035\}$.

(c) Over the ring~$\Z_9$ with the natural ordering and the reverse standard basis of $\Z_9^4$ we obtain the lexicode
      $C:=C(<,B,P)=\Z_9\{0003,0030,0300,3000\}$, which is not free and has~$81$ elements.
      If we reset $P[0]$ to false, we obtain the free lexicode $C':=\Z_9\{0114,1048\}$ of cardinality~$81$.
      From the identity $|C|\cdot|C^\perp|=9^4$ (see Remark~\ref{rem:Isotr}) we conclude that both codes are actually self-dual and thus
      maximally self-orthogonal.
\end{ex}
%%%%%%%%%%%%%%%%%%%

Over rings with even characteristic, Remark~\ref{rem:Isotr} is no longer sufficient, and
self-orthogonality cannot be described by a multiplicative property.
However, over the alphabet~$\Z_4$ it is known that if $C\subseteq\Z_4^n$ is a code such that
the Euclidean weight of each codeword $x\in C$ satisfies $\wtE(x)\equiv 0~\mod 8$ then~$C$ is self-orthogonal; see~\cite[Thm.~12.2.4]{HP03}.
These codes are known as self-orthogonal code of Type~II; see~\cite[p.~495]{HP03}.

All of this means that we can find self-orthogonal codes of Type~II in~$\Z_4^n$ using the multiplicative property
$[P[x]\Longleftrightarrow \wtE(x)\equiv 0~\mod 8]$.
This will in general not lead to maximal self-orthogonal codes because
self-dual codes of Type~II exist only if the length~$n$ is divisible by~$8$; see~\cite[Cor.~12.5.5]{HP03}.

%%%%%%%%%%%%%%%%%%%%%%%%%%%%
\begin{ex}\label{ex:SOZ4}
Consider~$\Z_4^5$ with the property $[P[x]\Longleftrightarrow \wt_E(x)\equiv 0~\mod 8]$ and the reverse standard basis $B=\{e_5,\ldots,e_1\}$.
Using the natural ordering $0<1<2<3$ on~$\Z_4$ we obtain the lexicode $C(<,B,P)=\Z_4\{00022,00202,02002,20002\}$, which
has cardinality~$16$.
It is clearly contained in the self-dual code $C=\Z_4\{00002,00020,00200,02000,20000\}$, which is of Type~I (i.e., not of Type~II).
Many more examples of self-orthogonal lexicodes over~$\Z_4$, including self-dual codes of Type~II of length~8,
are given in~\cite[Table~2]{GGS14}.
\end{ex}
%%%%%%%%%%%%%%%%%%%%%%%%%%

We briefly touch upon a selection property that arises in the context of DNA codes.

%%%%%%%%%%%%%%%%%%%%%%%%%%%%%
\begin{ex}
Consider~$\Z_4^4$ with the multiplicative property $[P[x]\Longleftrightarrow \wtU(x)\leq2]$; see Definition~\ref{D-Weights}(c).
Using the natural ordering~$<$ on~$\Z_4$ and the reverse standard basis~$B$ on~$\Z_4^4$, one obtains the lexicode
$C(<,B,P)=\Z_4\{0001,0010,0200,2000\}$, which has cardinality~$64$.
This idea could prove useful in constructing DNA codes with bounded GC-content, as discussed in~\cite{BGG15}, by suitably identifying the
elements of~$\Z_4$ with the~$4$ nucleotides $A,G,T,C$.
However, we wish to add that codes with \emph{constant} GC-content appear to be more useful for DNA computing as they guarantee a uniform hybridization process~\cite{MiKa06}.
These codes are clearly nonlinear and thus do not fall in the realm of this paper.
\end{ex}
%%%%%%%%%%%%%%%%%%%%%%%%%%%%%%%%%%%%%

We close the paper with an example over a noncommutative ring.
%%%%%%%%%%%%%%%%%%%%%%%%%%%%%%
\begin{ex}
Let $R=M_2(\F_2)$ be respectfully ordered as in Example \ref{E-RespOrder}(c).
Consider~$R^3$ with the reverse standard basis $B=\{e_3,e_2,e_1\}$, thus
\[
  e_1=(I_2,0,0),\ e_2=(0,I_2,0),\ e_3=(0,0,I_2).
\]
We use the selection property $[P[x]\Longleftrightarrow\rankSum(x)\geq 2]$, see Example~\ref{D-Weights}(b).
Then Algorithm~\ref{alg:Greedy} produces the lexicode
$C(<,B,P)=R\{(0,I_2,I_2),\,(I_2,0,I_2)\}$,
which is free of dimension~$2$ (as it has to be according to Theorem~\ref{thm:FalseFree}), thus cardinality~$256$.

\end{ex}
%%%%%%%%%%%%%%%%%%%%%%%%%%%%%%%%%%%%%%%

%%%%%%%%%%%%%%%%%%%%%%%%%%%%%%%%%%%%%%%%%%%%%%%
\bibliographystyle{abbrv}

\end{document}